%% file: Gossip.tex
\documentclass[submission,copyright,creativecommons]{eptcs}
\providecommand{\event}{FOMC 2013} 

\usepackage[english]{babel}
\usepackage{etex}
\usepackage{amsmath}%
\usepackage[amsmath,amsthm,thmmarks]{ntheorem}
\usepackage{amsfonts}%
\usepackage{graphicx}
\usepackage{babel}
\usepackage[normalem]{ulem}
\usepackage[usenames,dvipsnames]{pstricks}
\usepackage{epsfig}
\usepackage{pst-grad} 
\usepackage{pst-plot}
\usepackage{caption}
\usepackage{subcaption}
\usepackage{amssymb}
\usepackage{dsfont}
\usepackage{color}
\usepackage{multicol}
\usepackage[toc,page]{appendix}
\usepackage{times}

\makeatletter

\newcommand{\Rmnum}[1]{\expandafter\@slowromancap\romannumeral #1@}

\newcommand{\expectation}[1]{\mathds{E}\left[ {#1} \right] }

\newcommand{\probmath}[1]{\mathds{Pr}\left[ {#1} \right] }

\newcommand{\diam}[1]{d_{#1}}
\newcommand{\ds}{\mathcal{D}}
\newcommand{\dsopt}{\mathcal{D}_{opt}}

\newcommand{\bo}[0]{\mathcal{O}}
\newcommand{\gs}{Multi\_Broadcast}
\newcommand{\br}{Broadcast}
\newcommand{\maxdegree}{\Delta}
\newcommand{\maxpower}{\Phi}
\makeatother

\newcommand{\boldsubsection}[1]{%
  \par\addvspace{\baselineskip}
  \noindent\textbf{#1}\enspace\ignorespaces}

\newtheorem{theorem}{Theorem}

\newtheorem{claim}[theorem]{Claim}

\newtheorem{lemma}[theorem]{Lemma}

\usepackage[lined,boxed,commentsnumbered,linesnumbered,ruled]{algorithm2e}

 \ifx\proof\undefined\
   \newenvironment{proof}[1][\proofname]{\par
     \normalfont\topsep6\p@\@plus6\p@\relax
     \trivlist
     \itemindent\parindent
     \item[\hskip\labelsep
           \scshape
       #1]\ignorespaces
   }{%
     \endtrivlist\@endpefalse
   }
   \providecommand{\proofname}{Proof}
 \fi

\newcounter{mycount}
\newcommand\myprob[3]
{
  \stepcounter{mycount}
  \par\noindent
 {\bfseries Input:} #1\\
 {\bfseries Output:} #2 \\
}

\newcounter{mycount1}
\newcommand\probFormulation[2]
{
  \stepcounter{mycount1}
  \par\noindent
 {\bfseries Instance:} #1\\
 {\bfseries Question:} #2
}

\title{Message and time efficient multi-broadcast schemes}
\author{Liron Levin
\institute{Communication Systems Engineering Department, \\ Ben-Gurion University of the Negev, Israel\\}
\and
Dariusz R. Kowalski
\institute{Department of Computer Science, \\ University of Liverpool, UK}
\and
Michael Segal
\institute{Communication Systems Engineering Department, \\ Ben-Gurion University of the Negev, Israel\\}
}
\def\titlerunning{Message and time efficient multi-broadcast schemes}
\def\authorrunning{L. Levin, D. Kowalski \& M. Segal}

\begin{document}
\maketitle
\input{abstract.tex}
\input{introduction.tex}
\input{model.tex}
\input{PreviousWork.tex}
\input{ApproximationAlgorithm.tex}
\input{GossipGeneral.tex}
\input{ScheduleLength.tex}
\input{UnkownNetworksIntro.tex}
\input{DistributedImplementation.tex}
\input{UndefinedGraphs.tex}
\input{conclusions.tex}

\nocite{*}
\bibliographystyle{eptcs}
\bibliography{references}
\input{NPHardOmni.tex}
\end{document}

%% file: abstract.tex
\begin{abstract}
We consider message and time efficient broadcasting and multi-broadcasting in wireless ad-hoc networks,
where a subset of nodes, each with a unique rumor, wish to broadcast their rumors to all destinations while minimizing the
total number of transmissions and total time until all rumors arrive to their destination.
Under centralized settings, we introduce a novel approximation algorithm that provides almost optimal results with respect to the number of transmissions and total time, separately.
Later on, we show how to efficiently implement this algorithm under distributed settings, where the nodes
have only local information about their surroundings.
In addition, we show multiple approximation techniques based on the network collision detection capabilities and explain how
to calibrate the algorithms' parameters to produce optimal results for time and messages.
\end{abstract} 

%% file: introduction.tex
\section{Introduction}
Data \emph{broadcasting}, where a rumor from a single source has to be delivered to all other nodes in the graph, is considered one of the most studied problems in wireless ad-hoc networks \cite{BroadcastSurvey}.
In this paper, we study a generalized version called the \emph{multi-broadcast} problem~\cite{multibroadcast3}, where instead of a single source, a subset of sources $S \subseteq V$, each with a different rumor, have to deliver their rumors to {\itshape all} other nodes in the network.
When $S$ contains only a single node, the problem reduces to data \emph{broadcasting} problem, and when $S$ contains all the nodes, it reduces to data \emph{gossiping} problem~\cite{Gossip3}.

We use the partial aggregation model, also known as the \emph{combined message model}~\cite{Gossip1,combinedmodel2}, where a node can aggregate multiple messages to one by stripping message headers, using compression or correlating data from other nodes~\cite{funnling1}.
Formally, we use the compression factor $c$, which serve as an upper bound for the number of messages that can be
compressed to a single batch; note that a message can only be compressed once.
In this paper, we develop generalized algorithms which hold for any subset $S \subseteq V$ and and positive integer $c \in [1,k]$,
and thus suitable for both broadcasting and gossiping with and without aggregation (i.e., $c=1$).

In data dissemination, there are two important performance metrics that directly affect the quality of the algorithm: \emph{time} efficiency, measured by the total time until all nodes receive all rumors,
and \emph{message} efficiency, assessed by the total number of messages that are transmitted in the network.
Most papers on data broadcasting and gathering concentrate on optimizing the time metric \cite{Gossip1,Gossip2,Gossip3} and only provide by-product analysis of the message metric without exact performance guarantees.
However, In ad-hoc networks, where the nodes have limited battery and the cost of sending a message is directly proportional
to the lifetime of a node~\cite{MaximumLifetime}, minimizing the number of messages 
is a key aspect in the overall efficiency of the solution.
In this work, we concentrate on finding both message and time efficient algorithms
for broadcasting problem and for the more general multi-broadcasting problem, with and without aggregation.
We separate our analysis to two types of network settings: \emph{centralized} and \emph{distributed}.
In the \emph{centralized} network setting~\cite{Gossip2},
we assume that each node has full knowledge about the topology of the network, including
size, distance, and the ids of all nodes.
In the \emph{distributed} network settings~\cite{DistributedBroadcast1, DistributedBroadcast2, DistributedBroadcast3}, we assume that each node has only partial information about the network; for example, the number of neighbors it has or the total number of nodes.
\boldsubsection{Our results.}
For centralized network setting
we show a direct relation between messages efficiency and the size of the underlying \emph{backbone} topology,
on which rumors propagate to their destination,
and show how to build a backbone such that the number of message transmitted is small.
To handle time efficiency, we show how to shorten the diameter of the obtained backbone, which decrease the total time of the scheduling algorithm and ensures all rumors arrive to their destination as soon as possible. 
Our construction has minor impact on the message efficiency. 
Our results improves previous approximation ratio by Kim et al. ~\cite{BoundedDiameterCDS}.
For the distributed network settings, we first show how to construct the backbone on which rumors will propagate. 
Next, we show a message and time efficient technique for transmitting messages using the constructed backbone structure.
The technique enables calibrating the performance of the algorithm based on time or message
requirements.
The novelty of our approach is by comparing the quality of the proposed algorithms under each of the criteria, separately.
In addition, as a by-product of our work, we present an algorithm for building a connected dominating set with short diameter.

The rest of the paper is organized as follows: in Section~\ref{sec:model} we present the model of the network and formulate the multi-broadcast optimization problem.
Summary of related work is presented in Section~\ref{sec:previouswork}.
We provide approximation algorithms for efficient message and time
broadcast and multi-broadcast under centralized setting in Section~\ref{sec:centralized},
and extend this work for distributed setting in Section~\ref{sec:distributed}.
Our conclusions and future work are summarized in Section~\ref{sec:conclusions}.

%% file: model.tex
\section{Model and Problem Formulation}\label{sec:model}
Ad-hoc wireless networks consist of a set of $n$ mobile units, also called {\em nodes},
distributed in a two dimensional plane and equipped with radio transmitters and receivers. 
The power required to transmit a message from a node to distance $r$ is $P = r^{\alpha}$, where $r \in [1, \maxpower]$ is the transmission radius for some
physical system parameter $\maxpower>1$ and $\alpha \in [2,4]$ is the path loss exponent \cite{CommPhysics}. 
Our model is made more realistic by the incorporation of physical obstacles in the network, which represent buildings, trees or other objects that block message transmissions \cite{ObstacleAdHoc}.
The transmission power of each node $P(u)$ is pre-configured, and cannot change during the course of the algorithm,
and a directed edge between two nodes $u$ and $v$ is formed if there is no physical obstacle and if the Euclidean distance between them, $d(u,v)$
is less than $\sqrt[\alpha]{P_{u}}$.
In addition, we also consider the special case of Unit Disc Graphs (UDG), where $P(u)$ is equal for all nodes.

Let $k$ be the number of different rumors in the network.
In our model, the cost of sending a rumor from a node to its neighbors is fixed, but up to $c$ rumors,
$1 \leq c \leq k$, can be compressed to a single message, which we refer to as \emph{batch}.
Note that messages inside each batch can be rearranged in intermediate nodes
but multiple batches cannot be further compressed.
We consider the following parameters of the network graph $G$:
its diameter, $\diam{G}$, the degree of each node $\delta(v)_G$, its maximum degree $\maxdegree_G$ 
and $h_G(u,v)$, the shortest number of hops needed to route a message from $u$ to $v$ in $G$;
subscript $G$ is removed when it is clear from the context.

In this paper, we study the $\gs$ problem, which is defined as follows:
\myprob
{Graph $G=(V,E)$, set $S$ of $k$ source nodes each with one rumor, and compression parameter $c\le k$.}
{Multi-broadcast schedule from all nodes in $S$ to all nodes in $V$.}

For abbreviation we use $\br$ when $k = 1$ and $\gs$ otherwise.
Note that in some related work~\cite{Gossip1,Gossip2}, when $k = n$ the problem is refereed to as gossiping.

We are looking for a solution to the problem under the following optimization criterion:
\par\noindent
{\bfseries Message Efficiency:}
The objective here is to minimize the number of messages transmitted in the network in the course of the algorithm.
When analyzing {\itshape only} the message efficiency criteria, we do not take interferences into consideration, assuming that all messages can be scheduled by some interference-free protocol without
increasing the number of messages sent
(e.g., we can partition time into $n$ time slots, and let node $i$ transmit in time slot numbered $t = i \mod n$). 
This assumption is removed when additional optimization criterion are considered.
We define $m_{opt}$ as the minimum number of messages that are transmitted in the network
during the execution of the optimal solution.\\
\par\noindent
{\bfseries Time Efficiency:} 
The objective here is to minimize the time it takes until all rumors are received by all nodes.
When analyzing time efficiency, we adopt the \emph{protocol interference model}~\cite{PhysicalModelAndSinr},
where a communication between nodes $u$ and $v$ is successful if no neighbor of $v$ (the receiver) is simultaneously transmitting.
For any subgraph $T \subseteq G$, let $I_p(u,T)$ be the \emph{conflict set} of $u$ in $T$,
which consists of nodes that cannot be scheduled to transmit simultaneously with $u$ because they interfere
to $u$'s recipients.
Note that since we use omni-directional antennas we have $I_p(v,T) = \maxdegree_T (\maxdegree_T - 1)$.
We define $s_{opt}$ as the minimum time required to
deliver all rumors to their destinations. \\
\par\noindent

When analyzing the efficiency of an algorithm, the performance is compared to the optimal solution under each specified criteria, i.e., in time and message efficient $\gs$, we ask to find a schedule that uses at most $\alpha m_{opt}$ messages, and takes at most $\beta s_{opt}$ time, for some parameters $\alpha, \beta \geq 1$.
Approximation algorithms are used since time efficient gossip is NP-hard \cite{NpHardScheduling-GeneralGraph,NpHardScheduling-Udg} and, as we show in Appendix \ref{sec:NP-Hard}, message efficient gossip is also NP-hard.

To efficiently solve $\br$ under both centralized and distributed settings,
we also show how to construct an underlying graph on which rumors are routed from the source to the entire graph;
we refer to this graph as the network \emph{backbone}. 

%% file: PreviousWork.tex
\section{Previous Work}\label{sec:previouswork}
The problem of message and time efficient broadcast, multi-broadcast, and gossip $(k=n)$ has been studied in multiple research papers.
For centralized setting, Clementi et al.~\cite{Gossip1} studied gossip with
the existence of faulted  links in the networks.
They proposed an algorithm with time efficiency $\bo(n \maxdegree)$ and message efficiency $\bo(n^2)$
without compression, and time efficiency $\bo(\diam{G} \maxdegree)$ and message efficiency $\bo(\diam{G} n)$
with maximum compression, i.e., $c = k$.
The model was extended to include radio interference by \cite{Gossip2},
where Gasieniec et al. showed how to construct a data gathering tree for fast broadcast and gossip.
For broadcast, they proposed two algorithms, deterministic and randomized, with time efficiency
$\diam{G} + \bo(\log{n}^3)$ and
$\diam{G} + \bo(\log{n}^2)$, respectively.
It was later asymptotically improved by Kowalski and Pelc \cite{Kowalski:2007:ODB:1553848.1553850}, who provided a polynomial time deterministic algorithm computing broadcast protocol with time efficienct $\bo(\diam{G} + \log{n^2})$.
For gossip, Gasieniec et al. \cite{Gossip2} also showed an algorithm with time efficiency $\diam{G} + \maxdegree \cdot n + \bo(\log{n}^3)$.
If the maximum degree is bounded by $\Omega(\log{n})$, Cicalese et al.
\cite{FasterCentralizedCommunication} improved the result to
$\diam{G} + \bo(\frac{\log{n}^{3}}{\log{\log{n}}})$ time efficient broadcast and
$\bo(\diam{G} + \frac{\maxdegree \log{n}}{\log{\maxdegree} - \log{log{n}}})$ time efficient gossip.
In addition, they showed that their result is almost tight by constructing a $\maxdegree$-regular tree in which the time efficiency of gossip is at least $\Omega(\diam{G} + \frac{\maxdegree \log{n}}{\log{\maxdegree}})$.

For distributed setting, where the only information a node has is its coordinates,
Emek et al. \cite{DistributedBroadcast3} studied two initialization model in Unit Disk Graphs:
conditional wake up, when all stations other than the source are initially idle, and spontaneous
wake up, where all stations are initially awake.
They proposed a $\bo(\diam{G}g)$ and a $\max(\bo(\diam{G} + g^{2}),\diam{G} \log{g})$ time efficient algorithms,
where $g$ is the inverse of the minimum distance between any two nodes.
In addition, they showed that the lower bound for time efficiency for any deterministic algorithm is $\bo(\diam{G} \sqrt{g})$.
In a follow-up work \cite{Emek:2008:EDS:1400751.1400782} the authors showed that for grid networks,
although the lower bound for broadcast stays $\bo(\diam{G} \sqrt{g})$,
there is a faster algorithm with time efficiency $\bo(\diam{G} g^{\frac{5}{6}} \log{g})$.
For multi-broadcast, Chlebus et al. \cite{multibroadcast3} showed a distributing algorithm
that constructs a tree using an innovative breadth-then-depth traversal.
Their algorithm has time efficiency $\bo(k \log{n}^{3} + n \log{n}^{4})$
even for compression $c=1$.
Another interesting model was proposed by Chlebus et al. \cite{manytomany}, where only
sources exchange their messages.
Their algorithm has time efficiency $\bo(\max_{u,v \in s}d(u,v) + k + \log{n}^{2})$
for distributed settings with full message compression ($c=k$).
Both papers assumed that the network is directed.

Although its importance to network long-ability and overall energy consumption,
there have been limited research that emphasizes the importance of message efficiency in data routing.
In \cite{chelius:inria-00070343}, the authors have shown that multi-broadcasting problem is NP-hard
when trying to minimize time or message efficiency.
They did not provide NP-hard results for general compression ratio and did not present any algorithm for the problem.
The problem was also studied by Berenbrink et al.~\cite{MessageEfficientBroadcast},
where the gossip and broadcast in random and general networks were studied.
For random networks, a gossip algorithm where each node transmits at most $\bo(\log{n})$ messages was presented, and for general networks, a broadcast algorithm where each node transmits at most
$\bo(\frac{\log{n}^{2}}{\log{\frac{n}{\diam{G}}}})$ messages was given.

The algorithms in this paper use a modified version of the connected dominating set as a compact backbone for routing messages.
The minimum connected dominating set is NP-hard~\cite{NPHardGuide}, and approximable within the factor of $2+ H(\maxdegree)$, where $H(i)$ is the $i$-th harmonic number, for general graphs \cite{MCDSGeneral} and the factor of $7.8$ for Unit Disc Graphs~\cite{MCDSUDG}.
For ad-hoc networks, Kowalski and Jurdzinski \cite{wdagJurdzinskiK12} demonstrated how to construct a backbone (CDS) in SINR model with application to multi-broadcast. They also showed a construction in sublinear time for radio networks~\cite{opodisJurdzinskiK12}.
To the best of our knowledge, the only paper that made a connection between connected dominating set and message efficient gossiping
was \cite{Harary199715}, where Harary et al. demonstrated that the email gossip number,
which is the minimum number of messages required for gossiping messages from $n$ sources, is $n - 1 + |\dsopt|$, where
$\dsopt$ is the optimal minimum connected dominating set and $|\dsopt|$ is its size.
Another interesting sub-problem that we investigate is finding a small connected dominating set with short diameter.
For Unit Disk Graphs, Kim et al. \cite{BoundedDiameterCDS} found a connected dominating set with size $11.4 |\dsopt| + 6$ and with diameter $3 \diam{G} + 7$.
A related model was proposed by Du et al. \cite{virtualbackbone}, where the authors showed how
to construct a CDS such that given a parameter $\alpha$ the distance between two nodes $u$ and $v$ is at most  $\alpha \cdot d(u,v)$.
Their construction provides $H(\frac{\maxdegree (\maxdegree - 1)}{2})$-approximation to the size of the solution for all graphs $G^{'} \subseteq G$, such that $d_G^{'}(u,v) = d_G(u,v)$. 

%% file: ApproximationAlgorithm.tex
\section{Multi-Broadcast under centralized setting}\label{sec:centralized}
In this section, we show bi-criteria approximation algorithms for message and time efficient $\br$ and $\gs$,
i.e., our algorithms find a solution having cost within a factor of $\alpha$ from the optimal solution with minimum number of messages and within a factor of $\beta$ from the optimal solution that needs minimum time to distribute all rumors.
The proposed algorithms are designed for centralized networks, where each node has knowledge about the entire network topology.

\subsection{Approximation algorithm for message efficient $\br$}

In this subsection, we solve the message efficient
$\br$ problem. We show a relation between the number of messages needed for $\br$
and the minimum connected dominating set of the graph, and provide a constructive
algorithm that uses this fact for broadcasting.

\begin{algorithm}[H]

Find a connected dominating set $\ds$ in $G$.\\

Transmit the rumor from $s$ to all nodes using $\ds$ as a backbone.
\caption{Message efficient $\br$ \label{alg:broadcast}}
\end{algorithm}
\

Let $m_{opt}$ be the minimum number of messages required to complete broadcast, and $\dsopt$ be the size of the optimal dominating set.
Assume we have an $\alpha$-approximation algorithm for finding $\ds$ in Algorithm \ref{alg:broadcast}.
We claim the following:

\begin{lemma}\label{lemma:approx1}
Algorithm \ref{alg:broadcast} uses $\alpha m_{opt} + 1$ messages to broadcast the rumor from $s$.
\end{lemma}

\begin{proof}
Assume we have a solution for message efficient $\br$, and let $T$ be the connected subgraph
along which the source message is propagated from $s$ to all other nodes.
Let $l$ be the number of leaves in $T$, and $d$ be the number of internal nodes.
By definition, internal nodes in $T$ represent a connected dominating set in $G$.
In order to propagate the message to the entire tree, all internal nodes must transmit at least one message;
otherwise, we could transform the non-transmitting node to a leaf.
Thus, the total number of messages transmitted is at least:
\[
m = \left\{ \begin{array}{ll}
         |\ds| + 1 & \mbox{if s is a leaf} \\
         |\ds| & \mbox{otherwise}.\end{array}
		  \right.
\]
This is illustrated in Figure \ref{fig:mcdsToGossip}.
\input{MCDSSample.tex}
Clearly, we have $\dsopt \leq m_{opt}$. 
Therefore for any approximation algorithm, we
have $\alpha \cdot \dsopt \leq \alpha \cdot m_{opt}$.
Thus, we can use MCDS approximation algorithm, attach the source node
$s$ if it is not a part of the dominating set, and get a backbone on which we propagate the broadcast message.
The number of messages required is identical to the number of nodes in the solution and is at most $\alpha \cdot \dsopt + 1$.
\end{proof}

Combing this with the best known MCDS approximation algorithm yields the following lemma:
\begin{lemma}
Algorithm \ref{alg:broadcast} solves $\br$ and uses at most
$(H(\maxdegree) + 2) m_{opt} +  1$ messages for general graph and $7.8 m_{opt} + 1$ messages for UDG.
\end{lemma} 

%% file: MCDSSample.tex
\begin{figure}
\center
\begin{subfigure}[b]{0.42\textwidth}
\scalebox{1} 
{
\begin{pspicture}(0,-1.7)(4.42,1.7)
\definecolor{color185}{rgb}{0.2,0.2,1.0}
\psellipse[linewidth=0.04,dimen=outer](0.51,-0.01)(0.51,0.49)
\psellipse[linewidth=0.04,linecolor=color185,linestyle=dashed,dash=0.16cm 0.16cm,dimen=outer](2.27,1.21)(0.51,0.49)
\psellipse[linewidth=0.04,dimen=outer](3.91,0.01)(0.51,0.49)
\usefont{T1}{ptm}{m}{n}
\rput(0.47140625,-0.03){$s$}
\psline[linewidth=0.04cm](2.22,0.7)(2.22,-0.7)
\psline[linewidth=0.04cm](2.6,-1.06)(3.56,-0.36)
\psline[linewidth=0.04cm](0.74,0.4)(1.8,1.06)
\psellipse[linewidth=0.04,linecolor=color185,linestyle=dashed,dash=0.16cm 0.16cm,dimen=outer](2.19,-1.21)(0.51,0.49)
\end{pspicture}
}

\caption{$s$ is not part of MCDS, $|\dsopt| + 1 = m_{opt}$}
\label{fig:sampleWithoutSource}
\end{subfigure}%
\qquad \qquad \qquad
\begin{subfigure}[b]{0.40\textwidth}
\scalebox{1} 
{
\begin{pspicture}(0,-1.7)(4.42,1.7)
\definecolor{color185}{rgb}{0.2,0.2,1.0}
\psellipse[linewidth=0.04,dimen=outer](0.51,-0.01)(0.51,0.49)
\psellipse[linewidth=0.04,linecolor=color185,linestyle=dashed,dash=0.16cm 0.16cm,dimen=outer](2.27,1.21)(0.51,0.49)
\psellipse[linewidth=0.04,dimen=outer](3.91,0.01)(0.51,0.49)
\usefont{T1}{ptm}{m}{n}
\rput(2.1514063,-1.21){$s$}
\psline[linewidth=0.04cm](2.22,0.7)(2.22,-0.7)
\psline[linewidth=0.04cm](2.6,-1.06)(3.56,-0.36)
\psline[linewidth=0.04cm](0.74,0.4)(1.8,1.06)
\psellipse[linewidth=0.04,linecolor=color185,linestyle=dashed,dash=0.16cm 0.16cm,dimen=outer](2.19,-1.21)(0.51,0.49)
\end{pspicture}
}
\caption{$s$ is part of MCDS, $|\dsopt| = m_{opt}$}
\label{fig:sampleWithSource}
\end{subfigure}%
\caption{Similarity between MCDS and the optimal broadcast backbone.
The value of $m_{opt}$ denotes the size of optimal broadcast backbone, $|\dsopt|$ denotes the size of optimal minimum connected
dominating set and $s$ is the source.}
\label{fig:mcdsToGossip}					
\end{figure}
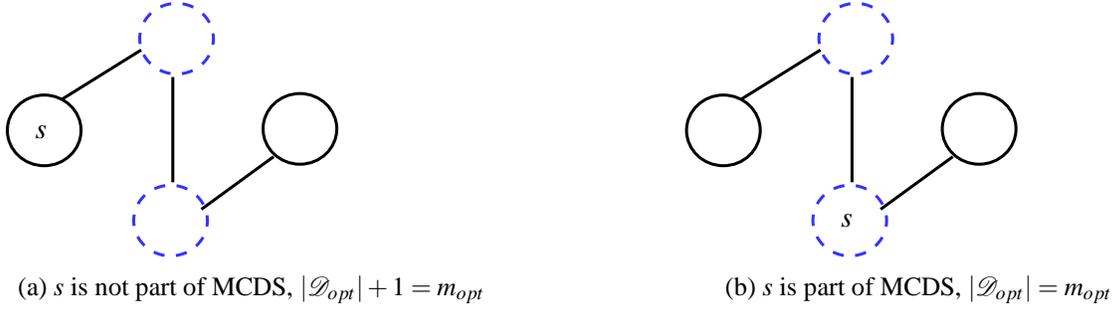 

%% file: GossipGeneral.tex
\subsection{Approximation algorithm for message efficient $\gs$}\label{Sec:GossipApprox}
In this subsection, we extend the result for message efficient $\br$ to $\gs$.
Note that we neglect the effect of interference since it is not part of the optimization criteria.
Let $|\dsopt|$ be the size of the MCDS in $G$, $m_{opt}$ be the minimum number of messages required to distribute all $k$ rumors in $\gs$, and
\[
|\dsopt^{s}| = \left\{ \begin{array}{ll}
         |\dsopt| - 1 & \mbox{if $s$ is a leaf} \\
         |\dsopt| & \mbox{otherwise.}\end{array}
		  \right.
\]

We claim the following:
\begin{lemma}
\label{lemma:lowerBound}
$\frac{k (|\dsopt|-1)}{c} \leq m_{opt}$.
\end{lemma}
\begin{proof}
Assume that $c=1$. We proved in Lemma \ref{lemma:approx1} that for any $s \in S$,
it is optimal to propagate node $s$ rumor using the MCDS.
Therefore, $m_{opt} \geq \sum_{s \in S}|\dsopt^{s}| \geq k \cdot \min_{s \in S}{|\dsopt^{s}|}$.

Now assume $c > 1$. For each $\dsopt^{s}$,
we cannot compress more than $c$ messages per node in $\dsopt^{s}$.
Therefore, every node $s$ is accountable for at least $\frac{|\dsopt^{s}|}{c}$ messages.
Our claim follows since:
\[
m_{opt} \geq \sum_{s \in S}\frac{|\dsopt^{s}|}{c} \geq
\frac{k \min_{s \in S}{|\dsopt^{s}|}}{c} \geq \frac{k (|\dsopt| -1 )}{c}.
\]
\end{proof}

Before introducing our main algorithm, recall that in centralized setting
each node knows in advance the structure of the entire topology.
Thus, every deterministic algorithm can be run inside each node without incurring additional messages.
Algorithm \ref{Algorithm-Approximation} is as follows:

\begin{algorithm}[H]
\KwIn{Graph $G = (V, E)$ and a set of source nodes $S$.}
\KwOut{A gossip schedule from each node $s \in S$ to all nodes in $V$.}

Find a connected dominating set $\ds$ in $G$\label{alg:findConnectedDominatingSet}.\\

Select the node with the lowest id as the root $r$. \newline
\tcc*[l]{For the next part, we assume we are provided with an interference aware protocol.}
Send the messages from all source nodes to $r$ over a path in $\ds$, aggregating messages when possible. \label{alg:sendMessageFromAllNodes}\\
Create a rooted arborescence from $r$, and
send all messages from $r$ to all nodes along $\ds$, sending exactly $\frac{k}{c}$ messages
by each transmitting node (internal and source). \label{alg:sendMessageFromR}
\caption{Message efficient $\gs$ \label{Algorithm-Approximation}}
\end{algorithm}
\

\begin{claim}\label{claim:allSourcesMustTransmit}
$k \leq m_{opt}$.
\end{claim}
\begin{proof}
The proof follows by the fact that every node $s \in S$ must transmit at least once.
\end{proof}

\begin{claim}\label{claim:SendingAllNodes}
Line \ref{alg:sendMessageFromAllNodes} in Algorithm \ref{Algorithm-Approximation} uses at most
$|\ds| \cdot \frac{k}{c} + k$ messages.
\end{claim}
\begin{proof}
Since we have $k$ sources, the maximum number of messages that any node belonging to $\ds$
may send is at most $\frac{k}{c}$, and we need to add the additional $k$ messages
to account the fact that some of the sources
may not belong to the dominating set $\ds$.
\end{proof}

\begin{claim}\label{claim:SendingFromRoot}
Line \ref{alg:sendMessageFromR} in Algorithm \ref{Algorithm-Approximation}
uses at most $(|\ds| + 1) \cdot \frac{k}{c}$ messages.
\end{claim}
\begin{proof}
For each internal node and for node $r$, we need $\frac{k}{c}$ messages to deliver all $k$ source messages to all of its neighbors.
Since we have $|\ds| + 1$ such nodes if $r$ belongs to $\ds$ and $|\ds|$ nodes otherwise, the claim follows.
\end{proof}

By combining Claims \ref{claim:allSourcesMustTransmit}, \ref{claim:SendingAllNodes} and \ref{claim:SendingFromRoot} we get that
the number of messages sent in Algorithm \ref{Algorithm-Approximation} is at most:
\[
(|\ds| + 1) \frac{k}{c} + |\ds| \frac{k}{c} + k  =
2 |\ds| \frac{k}{c} + k (1 + \frac{1}{c})
\]

We again use the $\alpha$-approximation algorithm for MCDS and obtain:
\[
2 ( \alpha |\dsopt|) \frac{k}{c} + k (1 + \frac{1}{c}) =
2 \alpha \frac{k (|\dsopt| - 1)}{c} + k(\frac{2 \alpha + 1}{c} + 1) \leq
\]
\[
2 \alpha m_{opt} + m_{opt} (\frac{2 \alpha + 1}{c} + 1) =
\]
\[
m_{opt} (2\alpha + 1)(1 + \frac{1}{c}).
\]
Using the approximation algorithm for MCDS from \cite{MCDSGeneral} for general graphs and
\cite{MCDSUDG} for UDG yields:
\begin{theorem}\label{theorem:approxwithouttime}
Algorithm \ref{Algorithm-Approximation}
is a $(2H(\maxdegree) + 5) (1+ \frac{1}{c})$-approximation algorithm for general graphs and $15.6 (1 + \frac{1}{c})$-approximation algorithm for UDG for message efficient $\gs$.
\end{theorem} 

%% file: ScheduleLength.tex
\subsection{Combining Time Efficiency}\label{sec:schedulinganalysis}

Algorithm \ref{alg:findConnectedDominatingSet} is optimized to reduce the message complexity of
the gossiping scheme.
However, the criterion of providing optimal results for time efficiency (i.e.,
minimizing the time until all $k$ rumors are received by all nodes) is still not satisfied.
Let $s_{opt}$ be the minimum time required to distribute all $k$ rumors to all nodes.
Clearly, $s_{opt} \geq \diam{G}$, since each rumor must propagate over the diameter.
Thus, the rumor distribution time of any algorithm that uses
a connected dominating set $\ds$ as a backbone is bounded by $\diam{\ds} + 2$ (e.g., when sender and receivers are leafs in $\ds$).
If $\diam{\ds}$ is relatively small, we can use one of the
interference aware scheduling algorithms on top of Algorithm~\ref{Algorithm-Approximation}, and produce a time and message efficient distribution scheme.
Therefore reducing the diameter of the resulting dominating set will necessarily improve the time efficiency of the algorithm.
Note that this task is not always trivial, since for some instances the diameter of the
minimum connected dominating set is $\bo(n)$ times the optimal diameter, 
see example Figure \ref{fig:mcdsScheduling}.

In this section, we first present Algorithm~\ref{Algorithm-CDS} for the problem of finding a minimum connected dominating set with bounded diameter, where we aim to find a dominating set $\ds$ of small size and small diameter,
with respect to $|\dsopt|$ and $\diam{G}$.
Once the backbone is constructed, we analyze the cost of scheduling messages over it and incorporate
the effect of interference.

\input{ScheduleLengthExp.tex}

\begin{algorithm}[H]
\KwIn{Graph $G = (V, E)$}
\KwOut{Connected dominating set $\ds$}

Find a connected dominating set $\ds$\label{line:findMCDS}.\\

Let $r$ be the node with minimum id in $\ds$.

Run a DFS traversal $DFS$, from $r$ on $\ds$ and partition $\ds$ to  clusters,
$C_{1},C_{2},...,C_{k}$,
each with diameter $\diam{G}$
(a node $v \in \ds$ will belong to cluster $j = \left\lfloor \frac{h_{dfs}(r,v)}{\diam{G}}\right\rfloor$;
a node that is traversed more than once can choose the dominating cluster arbitrarily). The value
$h_{dfs}(r,v)$ stands for the distance between $r$ and $v$ in DFS tree.\\ 

Let $c_i$ be the node with minimal id in cluster $C_i$; Set $c_i$ as the cluster leader.\\

Connect all cluster leaders to $r$ using shortest paths in graph $G$.\\

Add to set $\ds$ all nodes that belong to the found shortest paths.\\

\caption{Compute Minimum Diameter Connected Dominating Set\label{Algorithm-CDS}}
\end{algorithm}
\
We state the following.
\begin{lemma}
$ \diam{\ds} \leq 4 \cdot \diam{G}$.
\end{lemma}
\begin{proof}
Given $u \in C_{i}$ and $v \in C_j$, the length of the path from $u$ to $v$ is equal to
$h_{\ds}(u,c_{i}) + h_{\ds}(c_{i},r) + h_{\ds}(r,c_{j}) + h_{\ds}(c_{j},u) \leq 4 \cdot \diam{G}$, where $h_{\ds}(u,v)$ defines the length
of the path between $u$ and $v$ in $\ds$.
The last inequality holds since, by the definition of the diameter, for any $u,v \in V$, $h_{\ds}(u,v) \leq \diam{G}$.
\end{proof}

Assume we have an algorithm that finds a dominating set with size
$\alpha |\dsopt|$ in line \ref{line:findMCDS} of Algorithm \ref{Algorithm-CDS}, then we have:
\begin{lemma}
$|\ds| \leq 3 (\alpha |\dsopt|)$.
\end{lemma}
\begin{proof}
Since we perform a DFS traversal, each edge is traversed at most twice,
and the number of clusters is at most $2\frac{\alpha |\dsopt|}{\diam{G}}$.
The number of nodes that are added to connect each cluster is at most $\diam{G}$.
Thus, the new size of the dominating set is at most
$\alpha |\dsopt| + 2\frac{\alpha |\dsopt| }{\diam{G}} \cdot \diam{G} =
3 \alpha |\dsopt|$.
\end{proof}

To conclude, using \cite{MCDSGeneral} and \cite{MCDSUDG} for approximating MCDS in general and UDG networks, respectively, we get the following theorem:

\begin{theorem}\label{thm:scheduling}
Algorithm \ref{Algorithm-CDS} computes a connected dominating set of size at most
$3 (2 + H(\maxdegree)) |\dsopt|$
for general graphs and of size at most $15.6 |\dsopt|$ for unit disk graphs, and has
diameter of at most $4 \diam{G}$.
\end{theorem}

We are ready to analyze the time efficiency of Algorithm \ref{Algorithm-Approximation}
over the backbone constructed in Algorithm \ref{Algorithm-CDS}.
We begin by analyzing the algorithm under no interference assumption; later we will show how to incorporate
the interferences in the model.
We additionally assume that a node starts forwarding a batch of messages once it has $\frac{k}{c}$ messages.

The following lemma gives a lower bound for any $\gs$ algorithm:

\begin{claim}\label{claim:schedulingBound}
The time efficiency of any algorithm for $\gs$ is at least $\frac{k}{c} + \diam{G} - 1$.
\end{claim}
\begin{proof}
Consider a star with $k$ peripheral nodes and attach a path of
length $d$ to the center of the star.
Let $r$ be the farthest node from the center.
The optimal solution is to transmit all $k$ messages to the center simultaneously
and then directly send them to $r$.
This scheduling takes $\frac{k}{c} + \diam{G} - 1$ (without considering interference).
\end{proof}

For Algorithm \ref{Algorithm-CDS} we have the following upper bound:
\begin{lemma}\label{lem:schedulingefficiency}
Ignoring interferences,
the time efficiency of Algorithm \ref{Algorithm-Approximation} over any dominating set
is at most $2(4 s_{opt} + 1)$.
\end{lemma}
\begin{proof}
Assume we have a node $u$ with $k$ messages.
Ignoring interferences, sending a message from $u$ to $v \in V$ takes at least
$\frac{k}{c} - 1 + d(u,v)$, since $u$ has to transmit $\frac{k}{c}$ time to move all the messages
to the first hop neighbors, and additional $d(u,v) - 1$ time until the last message propagates from the first hop neighbors to $v$.
Optimally, $r$ belongs to $\ds$, and so $d(r,v)$ for the $k$ sources is at most $\diam{\ds} + 1$.
The algorithm is composed from two scheduling steps: converging all messages towards $r$ and disseminating them from $r$ over $\ds$.
Thus, we get that the time efficiency is:
\[
2(\frac{k}{c}-1 + \diam{\ds} + 1) = 2(\diam{\ds}+\frac{k}{c}) \leq 2(4\diam{G} + \frac{k}{c}) \leq 2(4 s_{opt} + 1).
\]
\end{proof}

Before incorporating interferences to the model, we state the following lemma:

\begin{lemma}[\cite{Levin2012}]\label{lem:schedulingTime}
Any deterministic scheduling algorithm on a spanning tree $T$ that ends after $t$ rounds can be transformed
to a collision free algorithm with $|I_p(u,T)| \cdot t$ rounds until completion.
\end{lemma}

Algorithm \ref{Algorithm-Approximation} uses spanning trees for the convergecast and broadcast operations,
so the algorithm from
\cite{Levin2012} can be used to perform the collision free scheduling, multiplying the scheduling time by $\max{|I_p(u,T)|}\leq \maxdegree_{\ds}^2$.

Combining
Lemma \ref{lem:schedulingefficiency},
Lemma \ref{lem:schedulingTime},
Theorem \ref{thm:scheduling},
and by changing the approximation ratio for MCDS from $\alpha$ to
$3 \alpha$ in Theorem \ref{theorem:approxwithouttime}
we get:
\begin{theorem}
Algorithm \ref{Algorithm-Approximation} on the connected dominating
set constructed by Algorithm \ref{Algorithm-CDS} has time efficiency
$\maxdegree_{\ds}^2 2(4s_{opt} + 1)$, message efficiency
$3  (2H(\maxdegree) +5)(1 + \frac{1}{c}) m_{opt}$ for general graphs and message efficiency $15.6 m_{opt}$ for unit disk graphs.
\end{theorem} 

%% file: ScheduleLengthExp.tex
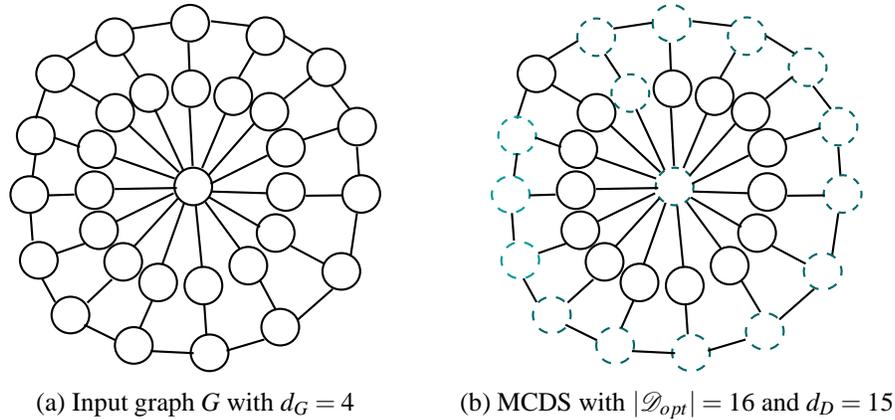
\begin{figure}
\centering
\begin{subfigure}[t]{0.4\textwidth}
\centering
\scalebox{0.7} 
{
\begin{pspicture}(0,-3.5)(7.06,3.5)
\psellipse[linewidth=0.04,dimen=outer](3.4912026,0.02875)(0.37829113,0.37625)
\psellipse[linewidth=0.04,dimen=outer](3.4512026,1.8875)(0.37829113,0.3675)
\psellipse[linewidth=0.04,dimen=outer](3.4312027,3.1325)(0.37829113,0.3675)
\psellipse[linewidth=0.04,dimen=outer](1.6139873,-0.04125)(0.37829113,0.37625)
\psellipse[linewidth=0.04,dimen=outer](2.0082912,1.42875)(0.37829113,0.37625)
\psellipse[linewidth=0.04,dimen=outer](2.1610758,-1.40375)(0.37829113,0.37625)
\psellipse[linewidth=0.04,dimen=outer](3.6839874,-1.85625)(0.37829113,0.37625)
\psellipse[linewidth=0.04,dimen=outer](4.536519,-1.47625)(0.37829113,0.37625)
\psellipse[linewidth=0.04,dimen=outer](4.936519,1.44875)(0.37829113,0.37625)
\psellipse[linewidth=0.04,dimen=outer](6.0170255,2.29375)(0.37829113,0.37625)
\psellipse[linewidth=0.04,dimen=outer](0.8685443,2.17)(0.37829113,0.3675)
\psellipse[linewidth=0.04,dimen=outer](0.37829113,-0.125)(0.37829113,0.3675)
\psellipse[linewidth=0.04,dimen=outer](1.1589241,-2.415)(0.37829113,0.3675)
\psellipse[linewidth=0.04,dimen=outer](3.8015823,-3.1325)(0.37829113,0.3675)
\psellipse[linewidth=0.04,dimen=outer](5.143481,-2.645)(0.37829113,0.3675)
\psline[linewidth=0.04cm](3.48,0.42)(3.42,1.52)
\psline[linewidth=0.04cm](3.66,0.38)(4.1,1.42)
\psline[linewidth=0.04cm](3.721519,-0.2675)(4.38,-1.16)
\psline[linewidth=0.04cm](3.54,-0.34)(3.657595,-1.515)
\psline[linewidth=0.04cm](3.32,-0.28)(2.9,-1.46)
\psline[linewidth=0.04cm](3.2,-0.2)(2.36,-1.12)
\psline[linewidth=0.04cm](1.74,1.62)(1.16,1.96)
\psline[linewidth=0.04cm](3.7,-2.2)(3.74,-2.78)
\psline[linewidth=0.04cm](4.76,-1.76)(5.06,-2.32)
\psline[linewidth=0.04cm](5.22,1.7)(5.74,2.12)
\psline[linewidth=0.04cm](3.42,2.22)(3.4,2.76)
\psline[linewidth=0.04cm](1.2332911,-0.125)(0.74,-0.14)
\psline[linewidth=0.04cm](1.0884811,2.45)(1.66,2.8)
\psline[linewidth=0.04cm](6.6,0.78)(6.64,0.24)
\psline[linewidth=0.04cm](6.12,1.94)(6.46,1.5)
\psline[linewidth=0.04cm](4.84,-2.88)(4.1222787,-3.0925)
\psline[linewidth=0.04cm](2.018481,-2.895)(1.42,-2.6)
\psline[linewidth=0.04cm](0.8686076,-2.17)(0.62,-1.72)
\psline[linewidth=0.04cm](0.5,1.32)(0.6837975,1.89)
\psline[linewidth=0.04cm](1.9483545,-1.7375)(1.48,-2.18)
\psellipse[linewidth=0.04,dimen=outer](5.2500634,-0.07625)(0.37829113,0.37625)
\psline[linewidth=0.04cm](4.9,-0.1)(3.8343039,-0.0675)
\psellipse[linewidth=0.04,dimen=outer](6.681709,-0.11375)(0.37829113,0.37625)
\psline[linewidth=0.04cm](6.34,-0.14)(5.6283545,-0.1475)
\psline[linewidth=0.04cm](6.6,-0.48)(6.48,-1.3)
\psellipse[linewidth=0.04,dimen=outer](5.2500634,0.76375)(0.37829113,0.37625)
\psellipse[linewidth=0.04,dimen=outer](6.601709,1.16625)(0.37829113,0.37625)
\psellipse[linewidth=0.04,dimen=outer](4.2500634,1.74375)(0.37829113,0.37625)
\psellipse[linewidth=0.04,dimen=outer](4.8617086,2.88625)(0.37829113,0.37625)
\psellipse[linewidth=0.04,dimen=outer](2.6500633,1.80375)(0.37829113,0.37625)
\psellipse[linewidth=0.04,dimen=outer](2.0017087,2.86625)(0.37829113,0.37625)
\psellipse[linewidth=0.04,dimen=outer](5.056519,-0.85625)(0.37829113,0.37625)
\psellipse[linewidth=0.04,dimen=outer](6.303481,-1.585)(0.37829113,0.3675)
\psline[linewidth=0.04cm](5.34,-1.1)(5.98,-1.44)
\psline[linewidth=0.04cm](3.5,-3.14)(2.68,-3.1)
\psline[linewidth=0.04cm](6.14,-1.92)(5.48,-2.5)
\psellipse[linewidth=0.04,dimen=outer](2.8239872,-1.79625)(0.37829113,0.37625)
\psellipse[linewidth=0.04,dimen=outer](2.3615823,-2.9925)(0.37829113,0.3675)
\psline[linewidth=0.04cm](2.72,-2.12)(2.48,-2.66)
\psellipse[linewidth=0.04,dimen=outer](1.7010759,-0.80375)(0.37829113,0.37625)
\psellipse[linewidth=0.04,dimen=outer](0.5589241,-1.375)(0.37829113,0.3675)
\psline[linewidth=0.04cm](0.44,-1.0)(0.40126583,-0.4775)
\psline[linewidth=0.04cm](1.4283544,-1.0575)(0.9,-1.24)
\psellipse[linewidth=0.04,dimen=outer](1.661076,0.73625)(0.37829113,0.37625)
\psellipse[linewidth=0.04,dimen=outer](0.5010759,0.99625)(0.37829113,0.37625)
\psline[linewidth=0.04cm](0.44,0.22)(0.5,0.66)
\psline[linewidth=0.04cm](2.328481,3.03)(3.06,3.14)
\psline[linewidth=0.04cm](3.828481,3.17)(4.58,3.08)
\psline[linewidth=0.04cm](5.228481,2.89)(5.76,2.56)
\psline[linewidth=0.04cm](5.58,0.92)(6.26,1.14)
\psline[linewidth=0.04cm](4.36,2.1)(4.64,2.58)
\psline[linewidth=0.04cm](1.34,0.78)(0.84,0.96)
\psline[linewidth=0.04cm](2.46,2.1)(2.1648102,2.545)
\psline[linewidth=0.04cm](3.3,0.34)(2.76,1.5)
\psline[linewidth=0.04cm](3.78,0.24)(4.66,1.24)
\psline[linewidth=0.04cm](3.82,0.08)(4.9,0.62)
\psline[linewidth=0.04cm](2.0,0.58)(3.14,0.18)
\psline[linewidth=0.04cm](2.0,-0.02)(3.18,0.0)
\psline[linewidth=0.04cm](2.02,-0.64)(3.12,-0.1)
\psline[linewidth=0.04cm](3.8,-0.18)(4.76,-0.74)
\psline[linewidth=0.04cm](3.22,0.3)(2.26,1.2)
\end{pspicture}
}
\caption{Input graph $G$ with $\diam{G}=4$}
\label{fig:ScheduleLengthAndMCDSInput}
\end{subfigure}%
\begin{subfigure}[t]{0.4\textwidth}
\centering
\scalebox{0.7} 
{
\begin{pspicture}(0,-3.5)(7.06,3.5)
\definecolor{color2706}{rgb}{0.0,0.4,0.4}
\definecolor{color2763}{rgb}{0.0,0.6,0.6}
\psellipse[linewidth=0.04,linecolor=color2706,linestyle=dashed,dash=0.16cm 0.16cm,dimen=outer](3.4912026,0.02875)(0.37829113,0.37625)
\psellipse[linewidth=0.04,dimen=outer](3.4512026,1.8875)(0.37829113,0.3675)
\psellipse[linewidth=0.04,linecolor=color2706,linestyle=dashed,dash=0.16cm 0.16cm,dimen=outer](3.4312027,3.1325)(0.37829113,0.3675)
\psellipse[linewidth=0.04,dimen=outer](1.6139873,-0.04125)(0.37829113,0.37625)
\psellipse[linewidth=0.04,dimen=outer](2.0082912,1.42875)(0.37829113,0.37625)
\psellipse[linewidth=0.04,dimen=outer](2.1610758,-1.40375)(0.37829113,0.37625)
\psellipse[linewidth=0.04,dimen=outer](3.6839874,-1.85625)(0.37829113,0.37625)
\psellipse[linewidth=0.04,dimen=outer](4.536519,-1.47625)(0.37829113,0.37625)
\psellipse[linewidth=0.04,dimen=outer](4.936519,1.44875)(0.37829113,0.37625)
\psellipse[linewidth=0.04,linecolor=color2706,linestyle=dashed,dash=0.16cm 0.16cm,dimen=outer](6.0170255,2.29375)(0.37829113,0.37625)
\psellipse[linewidth=0.04,dimen=outer](0.8685443,2.17)(0.37829113,0.3675)
\psellipse[linewidth=0.04,linecolor=color2763,linestyle=dashed,dash=0.16cm 0.16cm,dimen=outer](0.37829116,-0.125)(0.37829116,0.3675)
\psellipse[linewidth=0.04,linecolor=color2706,linestyle=dashed,dash=0.16cm 0.16cm,dimen=outer](1.1589241,-2.415)(0.37829113,0.3675)
\psellipse[linewidth=0.04,linecolor=color2706,linestyle=dashed,dash=0.16cm 0.16cm,dimen=outer](3.8015823,-3.1325)(0.37829113,0.3675)
\psellipse[linewidth=0.04,linecolor=color2706,linestyle=dashed,dash=0.16cm 0.16cm,dimen=outer](5.203481,-2.725)(0.37829113,0.3675)
\psline[linewidth=0.04cm](3.48,0.42)(3.42,1.52)
\psline[linewidth=0.04cm](3.66,0.38)(4.1,1.42)
\psline[linewidth=0.04cm](3.721519,-0.2675)(4.38,-1.16)
\psline[linewidth=0.04cm](3.54,-0.34)(3.657595,-1.515)
\psline[linewidth=0.04cm](3.32,-0.28)(2.9,-1.46)
\psline[linewidth=0.04cm](3.2,-0.2)(2.36,-1.12)
\psline[linewidth=0.04cm](1.74,1.62)(1.16,1.96)
\psline[linewidth=0.04cm](3.7,-2.2)(3.74,-2.78)
\psline[linewidth=0.04cm](4.76,-1.76)(5.12,-2.36)
\psline[linewidth=0.04cm](5.22,1.7)(5.74,2.12)
\psline[linewidth=0.04cm](3.42,2.22)(3.4,2.76)
\psline[linewidth=0.04cm](1.2332911,-0.125)(0.74,-0.14)
\psline[linewidth=0.04cm](1.0884811,2.45)(1.66,2.8)
\psline[linewidth=0.04cm](6.6,0.78)(6.64,0.24)
\psline[linewidth=0.04cm](6.12,1.94)(6.46,1.5)
\psline[linewidth=0.04cm](4.84,-2.88)(4.1222787,-3.0925)
\psline[linewidth=0.04cm](2.018481,-2.895)(1.42,-2.6)
\psline[linewidth=0.04cm](0.8686076,-2.17)(0.62,-1.72)
\psline[linewidth=0.04cm](0.5,1.32)(0.6837975,1.89)
\psline[linewidth=0.04cm](1.9483545,-1.7375)(1.46,-2.18)
\psellipse[linewidth=0.04,dimen=outer](5.2500634,-0.07625)(0.37829113,0.37625)
\psline[linewidth=0.04cm](4.9,-0.1)(3.8343039,-0.0675)
\psellipse[linewidth=0.04,linecolor=color2706,linestyle=dashed,dash=0.16cm 0.16cm,dimen=outer](6.681709,-0.11375)(0.37829113,0.37625)
\psline[linewidth=0.04cm](6.34,-0.14)(5.6283545,-0.1475)
\psline[linewidth=0.04cm](6.6,-0.48)(6.48,-1.3)
\psellipse[linewidth=0.04,dimen=outer](5.2500634,0.76375)(0.37829113,0.37625)
\psellipse[linewidth=0.04,linecolor=color2706,linestyle=dashed,dash=0.16cm 0.16cm,dimen=outer](6.601709,1.16625)(0.37829113,0.37625)
\psellipse[linewidth=0.04,dimen=outer](4.2500634,1.74375)(0.37829113,0.37625)
\psellipse[linewidth=0.04,linecolor=color2706,linestyle=dashed,dash=0.16cm 0.16cm,dimen=outer](4.8617086,2.88625)(0.37829113,0.37625)
\psellipse[linewidth=0.04,linecolor=color2706,linestyle=dashed,dash=0.16cm 0.16cm,dimen=outer](2.6500633,1.80375)(0.37829113,0.37625)
\psellipse[linewidth=0.04,linecolor=color2706,linestyle=dashed,dash=0.16cm 0.16cm,dimen=outer](2.0017087,2.86625)(0.37829113,0.37625)
\psellipse[linewidth=0.04,dimen=outer](5.056519,-0.85625)(0.37829113,0.37625)
\psellipse[linewidth=0.04,linecolor=color2706,linestyle=dashed,dash=0.16cm 0.16cm,dimen=outer](6.303481,-1.585)(0.37829113,0.3675)
\psline[linewidth=0.04cm](5.34,-1.1)(5.98,-1.44)
\psline[linewidth=0.04cm](3.5,-3.14)(2.68,-3.1)
\psline[linewidth=0.04cm](6.14,-1.92)(5.48,-2.5)
\psellipse[linewidth=0.04,dimen=outer](2.8239872,-1.79625)(0.37829113,0.37625)
\psellipse[linewidth=0.04,linecolor=color2706,linestyle=dashed,dash=0.16cm 0.16cm,dimen=outer](2.3615823,-2.9925)(0.37829113,0.3675)
\psline[linewidth=0.04cm](2.72,-2.12)(2.48,-2.66)
\psellipse[linewidth=0.04,dimen=outer](1.7010759,-0.80375)(0.37829113,0.37625)
\psellipse[linewidth=0.04,linecolor=color2763,linestyle=dashed,dash=0.16cm 0.16cm,dimen=outer](0.5589241,-1.375)(0.37829113,0.3675)
\psline[linewidth=0.04cm](0.44,-1.0)(0.40126583,-0.4775)
\psline[linewidth=0.04cm](1.4283544,-1.0575)(0.9,-1.24)
\psellipse[linewidth=0.04,dimen=outer](1.661076,0.73625)(0.37829113,0.37625)
\psellipse[linewidth=0.04,linecolor=color2763,linestyle=dashed,dash=0.16cm 0.16cm,dimen=outer](0.5010759,0.99625)(0.37829113,0.37625)
\psline[linewidth=0.04cm](0.44,0.22)(0.5,0.66)
\psline[linewidth=0.04cm](2.328481,3.03)(3.06,3.14)
\psline[linewidth=0.04cm](3.828481,3.17)(4.58,3.08)
\psline[linewidth=0.04cm](5.228481,2.89)(5.76,2.56)
\psline[linewidth=0.04cm](5.58,0.92)(6.26,1.14)
\psline[linewidth=0.04cm](4.36,2.1)(4.64,2.58)
\psline[linewidth=0.04cm](1.34,0.78)(0.84,0.96)
\psline[linewidth=0.04cm](2.46,2.1)(2.1648102,2.545)
\psline[linewidth=0.04cm](3.3,0.34)(2.76,1.5)
\psline[linewidth=0.04cm](3.78,0.24)(4.66,1.24)
\psline[linewidth=0.04cm](3.82,0.08)(4.9,0.62)
\psline[linewidth=0.04cm](2.0,0.58)(3.14,0.18)
\psline[linewidth=0.04cm](2.0,-0.02)(3.18,0.0)
\psline[linewidth=0.04cm](2.02,-0.64)(3.12,-0.1)
\psline[linewidth=0.04cm](3.8,-0.18)(4.76,-0.74)
\psline[linewidth=0.04cm](3.22,0.3)(2.26,1.2)
\end{pspicture}
}
\caption{MCDS with $|\dsopt|=16$ and $\diam{D}=15$}
\label{fig:MCDSWithNonOptimalScheduling}
\end{subfigure}%
\caption{Instance where the optimal MCDS does not yield optimal time efficiency.
The diameter of the input graph is always $4$, but the diameter of the chosen minimum dominating set
is equal to the size of the outer ring (and equals $\frac{n-1}{2}$).}
\label{fig:mcdsScheduling}					
\end{figure} 

%% file: UnkownNetworksIntro.tex
\section{Multi-Broadcast under distributed setting}\label{sec:distributed}
In this section, we focus on distributed network setting, where each node has only partial information about the network when the algorithm starts.
First, in Subsection \ref{sec:distributedimplmentation} we show how to distributively
construct the network backbone.
We emphasize that the efficiency of the construction is of less interest as we focus on finding a backbone
on which distributed multi-broadcasting is efficient with respect to time and messages.
Later, in Subsection \ref{sec:LocalInformation}, we show an efficient message and time scheduling routine,
which is used to route the rumors on top of the obtained backbone and explain how to 
apply the distributed versions of
Algorithm~\ref{Algorithm-Approximation} and Algorithm~\ref{Algorithm-CDS}.

%% file: DistributedImplementation.tex
\subsection{Constructing the network backbone in the interference-free settings}
\label{sec:distributedimplmentation}
In this step of developing distributed algorithms, we assume that all nodes have unique ids, and are aware of their local neighborhood (i.e., each node knows the
nodes that are in specific hop distance from it) and the diameter of the graph. 
In addition, the network is assumed to be signal interference-free.
The required steps to implement Algorithm \ref{Algorithm-Approximation} under distributed setting are to
construct a CDS, select a leader $r$, route all rumors from all $k$ sources to $r$ and
from $r$ to the entire graph along the found CDS.
For CDS construction, we use the distributed algorithm from \cite{Das97routingin},
which constructs a CDS with size $2H(\maxdegree) \dsopt$ using
$\bo(n |\diam{G}|)$ messages and $\bo(|\ds| (\maxdegree + \diam{\ds}))$ time for general graphs,
or the distributed algorithm from \cite{DistributedConstructionUDG}, which constructs a CDS with size at most
$8 \dsopt$ using $\bo(n \log  n)$ messages and $\bo(n)$ time for Unit Disc Graphs.
We use these
algorithms in such a way that a node can send a message batch once it has $c$ rumors in its
queue but not later than $\diam{G}$ time from the previous transmission.

The backbone constructed by Algorithm~\ref{Algorithm-CDS} extends the CDS
by dividing it to clusters using a depth-first-traversal, finding the shortest path from each cluster to $r$
and adding those paths to the constructed CDS.
After the CDS is constructed using one of the distributed routines,
we find $h_{dfs}(r,v)$ for each node $v \in \ds$ by selecting a leader $r$,
and running the distributed depth-first-traversal algorithm from \cite{DistributedDFS}.
In each cluster $C_{i}$, the nodes locally select a cluster leaders using the leader selection
algorithm from \cite{leaderselection} and discover the shortest path to $r$ using the routine from \cite{DistributedShortestPath}. All nodes in the shortest paths are then added to the CDS.
The message complexity of this step is $\bo(|E| + n \cdot \diam{G})$ and the time complexity is $\bo(n)$. 

%% file: UndefinedGraphs.tex
\subsection{Efficient message and time routine for spreading a rumor in the distributed
setting with interference}
\label{sec:LocalInformation}

In this subsection, we present multiple time and message efficient rumor distribution routines, which will later be used to implement $\gs$
using the distributed algorithm proposed in Subsection \ref{sec:distributedimplmentation}.
\paragraph{Setting with collision detection.}
Procedures \ref{proc:sendmessages} and \ref{proc:receivemessages} ensure that the number of messages
and scheduling time which
are needed to transmit a single rumor (or a compressed batch of rumors) from some node to its neighbors will not take too much time.
We assume that all nodes have synchronized clocks and have a collision detection mechanism.
Later we show how to modify the algorithm to support weaker scenario where,
in case of a collision, no signal is heard.
We also assume that each node $v$ is aware of the number of neighbors it has $\delta(v)$, and of the
maximum degree ($\maxdegree$).
This can be accomplished by
performing the neighbor discovery counting routine, using \cite{findNumberOfNeighbors},
which computes, w.h.p., a constant approximation degree of each node in
$\bo(\log^{2}{n})$ time using $\bo(\log{n} \cdot \delta(v))$ messages.

\begin{minipage}[t]{0.4\paperwidth}
\begin{procedure}[H]
Selects independently and uniformly at random an integer $x \in [1,\mu \maxdegree]$.\\
Send rumor in slot $x$. \\
Wait idle for $\mu \maxdegree + \mu \maxdegree -x$ slots. \\
\If{No error message arrive and no collision is heard in slot $\mu \maxdegree + x$}
{
Done
}
\Else
{
Collision occurred, retransmit
}
\caption{Send-Rumor() \label{proc:sendmessages}}
\end{procedure}
\end{minipage}
\hfill
\begin{minipage}[t]{0.4\paperwidth}
\begin{procedure}[H]
\If{Collision occurred in slot $x (x \leq \mu \maxdegree$) }
{
Send an error message in slot $\mu \maxdegree + x$.
}
\BlankLine \BlankLine \BlankLine \BlankLine \BlankLine \BlankLine \BlankLine \BlankLine \BlankLine \BlankLine \BlankLine \BlankLine \BlankLine \BlankLine \BlankLine \BlankLine \BlankLine \BlankLine \BlankLine \BlankLine \BlankLine \BlankLine
\caption{Receive-Rumor() \label{proc:receivemessages}}
\end{procedure}
\end{minipage}
Let $v$ be the neighbor of $u$,
$\mu \maxdegree$ be the number of slots,
$\delta(u)$ be the number of neighbors, and $X_u$ be the indicator variable for whether $v$ was the only neighbor of $u$ that
transmitted in slot $i$.
We have:
\[
\expectation{X_u}=\probmath{X_{u}}=
\binom{\mu \maxdegree}{1}\frac{1}{\mu \maxdegree}(1-\frac{1}{\mu \maxdegree})^{\delta(u)-1} \geq
(1-\frac{1}{\mu \maxdegree})^{\maxdegree-1} \geq
(1-\frac{1}{\mu \maxdegree})^{\maxdegree}.
\]
Therefore, the probability that all $\delta(v)$ neighbors will receive the rumor without interference is:
\[
\prod_{i=1}^{\delta(v)}(1-\frac{1}{\mu \maxdegree})^{\maxdegree} =
(1-\frac{1}{m})^{\maxdegree \delta(v)} \approx
e^{-\frac{\maxdegree}{\mu \maxdegree} \delta(v)} =
e^{-\frac{\delta(v)}{\mu}}.
\]
The probability that the transmission was unsuccessful for $j$ rounds is:
\[
(1-e^{-\frac{\delta(v)}{\mu}})^{j},
\]
and, thus, the expected number of retransmissions until all nodes receive the rumor is at most:
\[
\sum_{j=1}^{\infty} j (1-e^{-\frac{\delta(v)}{\mu}})^{j} \approx
e^{\frac{\delta(v)}{\mu}} (e^{\frac{\delta(v)}{\mu}}-1)
\leq
e^{2\frac{\delta(v)}{\mu}}.
\]
The expected number of error messages that all nodes send per round is:
\[
\delta(v)(1 - (1-\frac{1}{\mu \maxdegree})^{\maxdegree}) \approx
\delta(v)(1- e^{-\frac{1}{\mu}}).
\]

Let $\alpha$ be the approximation ratio of the distributed implementation of the minimum connected dominating set.
By applying the procedures for sending messages in the distributed implementation of
Algorithm \ref{Algorithm-Approximation} we get:
\begin{theorem}
For setting with collision detection, there is a distributed implementation of Algorithm \ref{Algorithm-Approximation}
for $\gs$
with message efficiency
$\bo(\alpha m_{opt} e^{2\frac{\delta(v)}{\mu}}(1+\delta(v)(1- e^{\frac{1}{\mu}}))$
and time efficiency \\
$\bo(s_{opt}\mu \maxdegree e^{2\frac{-\delta(v)}{\mu}})$.
\end{theorem}

\paragraph{Setting with no collision detection.}
As mentioned earlier, the collision detector restriction can be removed if, in each iteration,
a transmitting node selects only one of its neighbors as the receiver, and sends the rumor
directly to it in slot $x$.
If the rumor is received by the neighbor, it responds with an acknowledge message in slot $m + x$.
If no acknowledge message is received, the transmitting node assumes collision occurred and retransmits the message.
We have shown that the probability that a node will receive the message without collision is at most $e^{-\frac{1}{\mu}}$
and, thus,
the expected number of rounds until a rumor is received by all neighbors of node $v$ is at most
$e^{\frac{1}{\mu}} \delta(v)$.
Therefore, the time efficiency is at most $2 \mu \maxdegree e^{\frac{1}{\mu}} \delta(v)$
and the message efficiency is at most $2 e^{\frac{1}{\mu}} \delta(v)$.
To reduce the number of rounds, instead of iterating the above routine for all neighbors, $v$ can send a list of neighbors from which
he did not receive an acknowledge message.
In order to successfully receive a message, when a neighbor node $u$ chooses a slot, this
slot should not be taken by any neighbor of $u$ and the source node $v$.
Hence, the probability for a successful transmission is $e^{-\frac{2}{\mu}}$.
Let $j^{'}$ be the round after which all nodes received the rumor and
let $r(j)$ be the expected number of nodes that did not receive the rumor after round $j$.
We have,
$r(j)=\delta(v)(1-e^{\frac{-2}{\mu}})^{j}$.
Excluding the source, the expected number of nodes that transmit in round $j$ is at most the expected number of remaining nodes after round $j-1$, which is:
\[
\sum_{j=1}^{j^{'}}r(j-1) =
\sum_{j=1}^{j^{'}}\delta(v) (1-e^{-\frac{2}{\mu}})^{j-1} \leq
\sum_{j=1}^{\infty}\delta(v) (1-e^{-\frac{2}{\mu}})^{j} =
e^{\frac{2}{\mu}} \delta(v).
\]
The routine terminates when $r(j^{'}) < 1$ and, thus,
$j^{'} = \frac{\log{\frac{1}{\maxdegree}}}{\log{1-e^{-\frac{2}{\mu}}}}$.
Thus, the expected time efficiency is at most
$\mu \maxdegree j^{'}$.
By including the source, which transmits once per iteration, the expected message efficiency is at most $e^{\frac{2}{\mu}} \delta(v) + j^{'}$.
Combining with the distributed implementation of Algorithm \ref{Algorithm-Approximation} we get:
\begin{theorem}
For setting without collision detection, there is a distributed implementation of Algorithm \ref{Algorithm-Approximation}
for $\gs$
with message efficiency
$\bo(\alpha m_{opt} (e^{\frac{2}{\mu}} \delta(v) +
\frac{\log{\frac{1}{\maxdegree}}}{\log{1-e^{-\frac{2}{\mu}}}}))$
and time efficiency \\
$\bo(s_{opt}\mu \maxdegree \frac{\log{\frac{1}{\maxdegree}}}{\log{1-e^{-\frac{2}{\mu}}}})$.
\end{theorem} 

%% file: conclusions.tex
\section{Conclusions and future work.}\label{sec:conclusions}
In this work, we developed algorithms for message and time efficient
broadcasting and multi-broadcasting in ad-hoc networks under centralized and distributed settings.
We begin from a simple message efficient algorithm for centralized setting without considering interference and extend it to a complex interference-aware time and message efficient algorithm under distributed setting.
Future extension of our work could investigate the complexity of
multi-broadcasting under the SINR model or explore message and time efficiency algorithms when faulty links exist.
Another interesting question is to investigate the relation between the total time required for multi-broadcast, and the maximum degree of the underlying topology on which rumors propagate. 

%% file: NPHardOmni.tex
\begin{appendices}
\section{NP-Completeness}\label{sec:NP-Hard}
In this section, we show that both message efficient $\br$ and $\gs$ are NP-Complete under centralized setting.
We use a reduction from SET-COVER, where we are given a collection of sets $U_1,....,U_n$, whose union is $U$, and a natural number $l$.
We then ask if there is a sub-collection of at most $l$ sets whose union is $U$.
To show the reduction from $\br$ and $\gs$ to SET-COVER, we consider the following decision problem:
\probFormulation
{A graph $G=(V,E)$, a set of source nodes $S=\left\{s_1,....,s_i\right\}$, and $m \in \mathbb{N}$.}
{Is there a broadcast/gossip scheme that uses at most $m$ messages?}

\begin{lemma}
\label{proof:SingleBroadcastNpHard}
Message efficient $\br$ is NP-Complete.
\end{lemma}

\begin{proof}
We show that SET-COVER $\leq_{p} \br$ by constructing a $3$-tier graph from the SET-COVER instance, where the nodes in the second tier represent the sets, and the nodes in the third tier represent the elements.
The construction is as follows: create a source node $s$ with one rumor to distribute, and place it in the first tier.
In the second tier, create one node $i$ for each set $U_i$, and in the third tier create one node per element in set $U$.
Connect $s$ to all nodes in the second tier, and each node in the second tier to the nodes from the third tier that are associated with its elements (see Figure \ref{fig:reduction}).
To support such network construction in Euclidean plane,
we set equal transmission power $P_u$ for all nodes $v$ in the second tier, such that $P_{u} \gg d(u,v)^{\alpha}$ for any node $v$,\footnote{Here the symbol ``$\gg$'' means that the left hand side is bigger by
the right hand side by sufficiently large positive (constant) factor.}
and a very small power $P_w$ for all nodes $w$ in the third tier, such that
$d(u,v)^{\alpha}\gg P_w$ for any node $v$.
By doing so, an edge is created between any pair of nodes from the second tier and the third tier.
To prevent the propagation of messages from node $i$ to nodes $v$ that are not associated with $U_i$ we simply surround $v$ with obstacles.
As for the source node $s$, we place it sufficiently far from the second tier so that no node from
the second or the third tier could reach it, and associate with it a sufficient power to reach all nodes
in the second tier.

Since we use omnidirectional antennas, after one transmission, the message from $s$ arrives to all nodes in the second tier, and when a node from the second tier transmits, all connected nodes from the third tier receive the message.
We then ask, if there is a solution to $\br$ that uses at most $m=l+1$ messages. Clearly, in this solution, at most $l$ nodes from the second tier transmit, and all nodes receive the message from $s$.
To construct a SET-COVER with $l$ sets from $\br$ with $m=l+1$ messages, we select the sets that represent the $l$ nodes that transmitted.
Since there are $l$ nodes and all elements are covered by those nodes, they form a SET-COVER of $U$.
\end{proof}

\begin{lemma}
Message efficient $\gs$ is NP-Complete.
\end{lemma}

\begin{proof}
We use similar construction as in Lemma \ref{proof:SingleBroadcastNpHard}, but instead of one transmitting node, we create a clique of size $k$
in the first tier
and connect each node in the clique to all second tier nodes (see Figure \ref{fig:reductionGossip}).
We then ask whether there is a solution to $\gs$ with $m=k \cdot \lceil \frac{k}{c} \rceil +k$ messages.
Clearly, in any optimal algorithm, all $k$ rumors must be delivered to the intermediate nodes for future distribution and it is not possible to aggregate them.
Once we have all $k$ rumors in the intermediate nodes, once we can find a schedule such that only $k$ intermediate nodes transmit, we can use the same arguments as in
the proof of Lemma \ref{proof:SingleBroadcastNpHard} and show that SET-COVER $\leq_{p} \gs$.
\end{proof}
\input{reduction.tex}
\end{appendices} 

%% file: reduction.tex
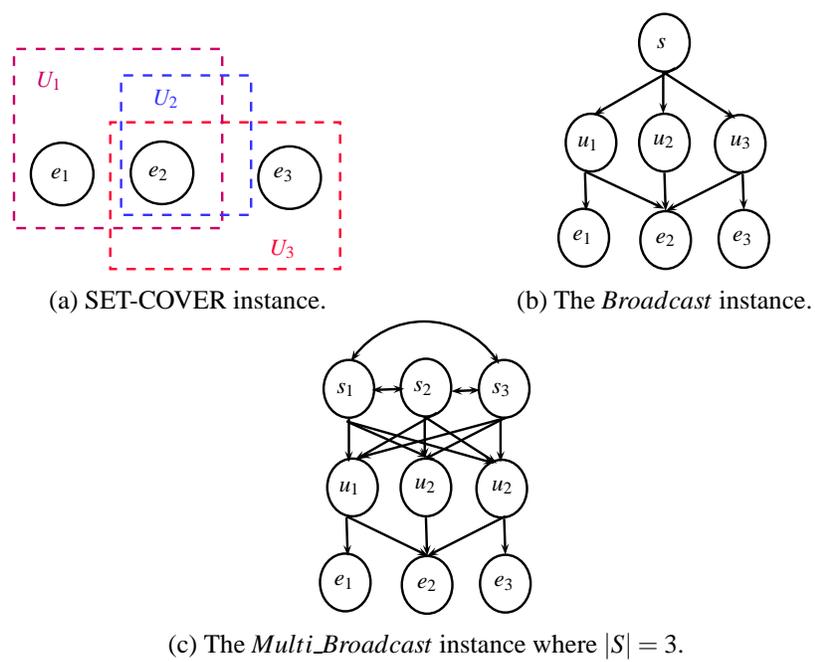
\begin{figure}
\center
\begin{subfigure}[b]{0.3\textwidth}
\scalebox{0.8}
{
\begin{pspicture}(0,-1.85)(5.46,1.85)
\definecolor{color556}{rgb}{1.0,0.0,0.2}
\definecolor{color582}{rgb}{0.8,0.0,0.4}
\definecolor{color659}{rgb}{0.2,0.2,1.0}
\psellipse[linewidth=0.04,dimen=outer](0.82,-0.25)(0.54,0.54)
\psellipse[linewidth=0.04,dimen=outer](2.48,-0.23)(0.54,0.54)
\psellipse[linewidth=0.04,dimen=outer](4.6,-0.31)(0.54,0.54)
\usefont{T1}{ptm}{m}{n}
\rput(0.80140626,-0.26){$e_1$}
\usefont{T1}{ptm}{m}{n}
\rput(2.4214063,-0.24){$e_2$}
\usefont{T1}{ptm}{m}{n}
\rput(4.501406,-0.26){$e_3$}
\psframe[linewidth=0.04,linecolor=color556,linestyle=dashed,dash=0.16cm 0.16cm,dimen=outer](5.46,0.63)(1.6,-1.85)
\psframe[linewidth=0.04,linecolor=color582,linestyle=dashed,dash=0.16cm 0.16cm,dimen=outer](3.5,1.85)(0.0,-1.17)
\usefont{T1}{ptm}{m}{n}
\rput(4.4914064,-1.5){\color{color556}$U_3$}
\usefont{T1}{ptm}{m}{n}
\rput(0.61140627,1.3){\color{color582}$U_1$}
\psframe[linewidth=0.04,linecolor=color659,linestyle=dashed,dash=0.16cm 0.16cm,dimen=outer](3.98,1.41)(1.78,-0.95)
\usefont{T1}{ptm}{m}{n}
\rput(2.5514061,1.0){\color{color659}$U_2$}
\end{pspicture}
}

\caption{SET-COVER instance.}					
\label{fig:reductionInput}
\end{subfigure}%
\qquad \qquad 
\begin{subfigure}[b]{0.3\textwidth}
\centering
\scalebox{0.8} 
{
\begin{pspicture}(0,-2.14)(3.561875,2.14)
\psellipse[linewidth=0.04,dimen=outer](0.56,-0.02)(0.44,0.5)
\psellipse[linewidth=0.04,dimen=outer](1.78,1.64)(0.44,0.5)
\psellipse[linewidth=0.04,dimen=outer](3.04,-0.04)(0.44,0.5)
\psellipse[linewidth=0.04,dimen=outer](1.78,-0.02)(0.44,0.5)
\psellipse[linewidth=0.04,dimen=outer](0.44,-1.6)(0.44,0.5)
\psellipse[linewidth=0.04,dimen=outer](3.1,-1.62)(0.44,0.5)
\psellipse[linewidth=0.04,dimen=outer](1.8,-1.64)(0.44,0.5)
\psline[linewidth=0.04cm,arrowsize=0.05291667cm 2.0,arrowlength=1.4,arrowinset=0.4]{<-}(0.6,0.44)(1.96,1.22)
\psline[linewidth=0.04cm,arrowsize=0.05291667cm 2.0,arrowlength=1.4,arrowinset=0.4]{<-}(1.76,0.46)(1.76,1.16)
\psline[linewidth=0.04cm,arrowsize=0.05291667cm 2.0,arrowlength=1.4,arrowinset=0.4]{<-}(2.96,0.38)(1.78,1.14)
\psline[linewidth=0.04cm,arrowsize=0.05291667cm 2.0,arrowlength=1.4,arrowinset=0.4]{<-}(3.1,-1.14)(3.08,-0.54)
\psline[linewidth=0.04cm,arrowsize=0.05291667cm 2.0,arrowlength=1.4,arrowinset=0.4]{<-}(1.8,-1.16)(1.78,-0.52)
\psline[linewidth=0.04cm,arrowsize=0.05291667cm 2.0,arrowlength=1.4,arrowinset=0.4]{<-}(0.48,-1.16)(0.46,-0.52)
\usefont{T1}{ptm}{m}{n}
\rput(0.51140624,-0.03){$u_1$}
\usefont{T1}{ptm}{m}{n}
\rput(1.7714063,0.01){$u_2$}
\usefont{T1}{ptm}{m}{n}
\rput(3.0514061,-0.01){$u_3$}
\psline[linewidth=0.04cm,arrowsize=0.05291667cm 2.0,arrowlength=1.4,arrowinset=0.4]{<-}(1.8,-1.16)(3.1,-0.5)
\psline[linewidth=0.04cm,arrowsize=0.05291667cm 2.0,arrowlength=1.4,arrowinset=0.4]{<-}(1.78,-1.12)(0.46,-0.5)
\usefont{T1}{ptm}{m}{n}
\rput(0.42140624,-1.59){$e_1$}
\usefont{T1}{ptm}{m}{n}
\rput(1.8014063,-1.63){$e_2$}
\usefont{T1}{ptm}{m}{n}
\rput(3.0814064,-1.61){$e_3$}
\usefont{T1}{ptm}{m}{n}
\rput(1.7314062,1.63){$s$}
\end{pspicture}
}

\caption{The $\br$ instance.}
\label{fig:reductionOutput}
\end{subfigure}
\qquad \qquad 
\begin{subfigure}[b]{0.5\textwidth}
\centering
\scalebox{0.8} 
{
\begin{pspicture}(0,-2.44)(3.561875,2.46)
\psellipse[linewidth=0.04,dimen=outer](0.56,-0.32)(0.44,0.5)
\psellipse[linewidth=0.04,dimen=outer](1.78,1.34)(0.44,0.5)
\psellipse[linewidth=0.04,dimen=outer](3.04,-0.34)(0.44,0.5)
\psellipse[linewidth=0.04,dimen=outer](1.78,-0.32)(0.44,0.5)
\psellipse[linewidth=0.04,dimen=outer](0.44,-1.9)(0.44,0.5)
\psellipse[linewidth=0.04,dimen=outer](3.1,-1.92)(0.44,0.5)
\psellipse[linewidth=0.04,dimen=outer](1.8,-1.94)(0.44,0.5)
\psline[linewidth=0.04cm,arrowsize=0.05291667cm 2.0,arrowlength=1.4,arrowinset=0.4]{<-}(0.6,0.14)(1.96,0.92)
\psline[linewidth=0.04cm,arrowsize=0.05291667cm 2.0,arrowlength=1.4,arrowinset=0.4]{<-}(1.76,0.16)(1.76,0.86)
\psline[linewidth=0.04cm,arrowsize=0.05291667cm 2.0,arrowlength=1.4,arrowinset=0.4]{<-}(2.96,0.08)(1.78,0.84)
\psline[linewidth=0.04cm,arrowsize=0.05291667cm 2.0,arrowlength=1.4,arrowinset=0.4]{<-}(3.1,-1.44)(3.08,-0.84)
\psline[linewidth=0.04cm,arrowsize=0.05291667cm 2.0,arrowlength=1.4,arrowinset=0.4]{<-}(1.8,-1.46)(1.78,-0.82)
\psline[linewidth=0.04cm,arrowsize=0.05291667cm 2.0,arrowlength=1.4,arrowinset=0.4]{<-}(0.48,-1.46)(0.46,-0.82)
\usefont{T1}{ptm}{m}{n}
\rput(0.51140624,-0.33){$u_1$}
\usefont{T1}{ptm}{m}{n}
\rput(1.7714063,-0.29){$u_2$}
\usefont{T1}{ptm}{m}{n}
\rput(3.0514061,-0.31){$u_2$}
\psline[linewidth=0.04cm,arrowsize=0.05291667cm 2.0,arrowlength=1.4,arrowinset=0.4]{<-}(1.8,-1.46)(3.1,-0.8)
\psline[linewidth=0.04cm,arrowsize=0.05291667cm 2.0,arrowlength=1.4,arrowinset=0.4]{<-}(1.82,-1.46)(0.46,-0.8)
\usefont{T1}{ptm}{m}{n}
\rput(0.42140624,-1.89){$e_1$}
\usefont{T1}{ptm}{m}{n}
\rput(1.8014063,-1.93){$e_2$}
\usefont{T1}{ptm}{m}{n}
\rput(3.0814064,-1.91){$e_3$}
\usefont{T1}{ptm}{m}{n}
\rput(1.7314062,1.33){$s_2$}
\psellipse[linewidth=0.04,dimen=outer](0.5,1.32)(0.44,0.5)
\usefont{T1}{ptm}{m}{n}
\rput(0.45140624,1.31){$s_1$}
\psellipse[linewidth=0.04,dimen=outer](3.08,1.32)(0.44,0.5)
\usefont{T1}{ptm}{m}{n}
\rput(3.0314062,1.31){$s_3$}
\psline[linewidth=0.04cm,arrowsize=0.05291667cm 2.0,arrowlength=1.4,arrowinset=0.4]{<-}(0.5,0.1)(0.5,0.82)
\psline[linewidth=0.04cm,arrowsize=0.05291667cm 2.0,arrowlength=1.4,arrowinset=0.4]{<-}(2.92,0.08)(0.46,0.78)
\psline[linewidth=0.04cm,arrowsize=0.05291667cm 2.0,arrowlength=1.4,arrowinset=0.4]{<-}(3.02,0.1)(3.02,0.82)
\psline[linewidth=0.04cm,arrowsize=0.05291667cm 2.0,arrowlength=1.4,arrowinset=0.4]{<-}(0.68,0.18)(3.12,0.84)
\psline[linewidth=0.04cm,arrowsize=0.05291667cm 2.0,arrowlength=1.4,arrowinset=0.4]{<-}(1.74,0.14)(3.1,0.84)
\psline[linewidth=0.04cm,arrowsize=0.05291667cm 2.0,arrowlength=1.4,arrowinset=0.4]{<-}(1.82,0.18)(0.5,0.76)
\psline[linewidth=0.04cm,arrowsize=0.05291667cm 2.0,arrowlength=1.4,arrowinset=0.4]{<->}(2.2,1.28)(2.68,1.3)
\psline[linewidth=0.04cm,arrowsize=0.05291667cm 2.0,arrowlength=1.4,arrowinset=0.4]{<->}(0.88,1.3)(1.42,1.32)
\psarc[linewidth=0.04,arrowsize=0.05291667cm 2.0,arrowlength=1.4,arrowinset=0.4]{<->}(1.75,0.99){1.45}{30.67728}{147.09476}
\end{pspicture}
}
\caption{The $\gs$ instance where $|S|=3$.}
\label{fig:reductionGossip}
\end{subfigure}

\caption{The reduction from SET-COVER to $\br$ and $\gs$.}
\label{fig:reduction}					
\end{figure} 